\documentclass{amsart}

\usepackage{enumerate}

\usepackage[T1]{fontenc}
\usepackage[utf8]{inputenc}

\usepackage{amsmath}
\usepackage{amssymb}
\usepackage{amsthm}
\usepackage{color}
\usepackage{xypic}

\theoremstyle{plain}
\newtheorem{theorem}{Theorem}
\newtheorem*{theorem*}{Theorem}
\newtheorem{proposition}[theorem]{Proposition}
\newtheorem{corollary}[theorem]{Corollary}
\newtheorem*{corollary*}{Corollary}
\newtheorem{lemma}[theorem]{Lemma}

\theoremstyle{remark}
\newtheorem{remark}[theorem]{Remark}
\newtheorem{example}[theorem]{Example}

\theoremstyle{definition}
\newtheorem{definition}[theorem]{Definition}

\newcommand{\R}{\mathcal{R}}
\newcommand{\field}[1][]{\mathbb{F}_{#1}}
\newcommand{\trace}{\operatorname{Tr}}
\newcommand{\Aut}[2][]{\operatorname{Aut}_{#1}\left(#2\right)}
\newcommand{\matrixring}[2]{\mathcal{M}_{#1}(#2)}
\newcommand{\tovector}{\mathfrak{v}}

\newcommand{\tomatrix}{\mathfrak{m}}
\newcommand{\transpose}[1]{#1^{\mathtt{T}}}
\newcommand{\kronecker}{\boxtimes}
\newcommand{\identity}[1]{\operatorname{id}_{#1}}

\newcommand{\norm}[2]{N_{#1}(#2)}
\newcommand{\lclm}[1]{\left[#1\right]_\ell}
\newcommand{\gcrd}[1]{\left(#1\right)_r}
\newcommand{\lcrm}[1]{\left[#1\right]_r}
\newcommand{\gcld}[1]{\left(#1\right)_\ell}
\newcommand{\lann}[2]{\mathsf{l}_{#1}(#2)}
\newcommand{\rann}[2]{\mathsf{r}_{#1}(#2)}
\newcommand{\rOre}[3]{[#2;#3]#1}
\newcommand{\lOre}[3]{#1[#2;#3]}

\title{Dual skew codes from annihilators: Transpose Hamming ring extensions}
\author[J. G\'omez-Torrecillas]{Jos\'e G\'omez-Torrecillas}
\email{gomezj@ugr.es}
\address{CITIC and Department of Algebra, University of Granada, Spain}
\author{F. J. Lobillo}
\email{jlobillo@ugr.es}  
\address{CITIC and Department of Algebra, University of Granada, Spain}
\author[G. Navarro]{Gabriel Navarro}
\email{gnavarro@ugr.es}
\address{CITIC and Department of Computer Sciences and AI, University of Granada, Spain}

\thanks{Research supported by grants  MTM2013-41992-P
from Ministerio de Econom\'{\i}a y Competitividad and FEDER, and MTM2016-78364-P from Agencia Estatal de Investigaci\'on and FEDER}

\begin{document}
\maketitle

\section*{Introduction}

Linear codes may be endowed with cyclic structures by means of skew polynomial rings. This is the case of Piret cyclic convolutional codes \cite{Piret:1976} and the subsequent generalizations and alternatives (see \cite{Roos:1979}, \cite{Gluesing/Schmale:2004}, \cite{Estrada/alt:2008}, \cite{LopezPermouth/Szabo:2013}, \cite{GLN:2016}, \cite{gln2016new}). Non commutative cyclic structures of this kind have been also considered for block linear codes  (\cite{Boucher/Ulmer:2009}, \cite{Boucher/alt:2007}, \cite{Boucher/Ulmer:2011}, \cite{Fogarty/Gluesing:2015}, \cite{Alhamadi/alt:2016}), and for linear codes over commutative rings (\cite{Boucher/alt:2008}, \cite{Jitman/alt:2012}, \cite{Ducoat/Oggier:2016}). 

A desirable property of any class of linear codes is to be stable under duals. This property has been already studied in several of the aforementioned references. A common feature of many of these approaches to duality is the presence of a suitable anti-isomorphism of rings that encodes, more or less explicitly, the transfer of the cyclic structure from the code to its dual. Our aim is to present a systematization of this method, besides some relevant examples where it successfully applies. 

The strategy is to establish a formal framework, called transpose Hamming ring extension, designed to derive that the dual of every cyclic code is cyclic. Cyclic codes will be, from an algebraic point of view, identified as left ideals of suitable (non-commutative) ring extensions of a given commutative ring $C$, well understood that such an ``identification'' has to be made explicit by an isomorphism of $C$--modules from the ring to $C^m$, where $m$ is the length of the $C$--linear code. The transposition will be an anti-isomorphism of rings which allows to transform annihilators into duals. Details are to be found in Section \ref{DualAnn}

In Section \ref{IdealCodes} we apply our general approach to left ideal convolutional codes in the sense of \cite{LopezPermouth/Szabo:2013}, extending to a more general setting, and improving, results from \cite{LopezPermouth/Szabo:2013} and \cite{GLN:2016} on the description of dual codes in this setting (Theorem \ref{dualidealcode}). The case of a simple word-ambient algebra is analyzed in detail (Theorem \ref{matrixmain}).  

Section \ref{sec:constacyclic} is devoted to dual codes of skew constacyclic codes over a commutative ring. Several results from \cite{Boucher/Ulmer:2009}, \cite{Boucher/alt:2008}, \cite{Jitman/alt:2012}, \cite{Ducoat/Oggier:2016} on these codes are covered by our general result (Theorem \ref{dualconstacyclic}).  

Finally, in Section \ref{DualsRS}, we compute (Theorem \ref{sRSdual}) the dual of a skew Reed Solomon code over a general field in the sense of \cite{GLNPGZ}, and, as a consequence, these codes are shown to be evaluation codes (Theorem \ref{sRSevaluation}).

\section{Transpose ring extensions and dual codes}\label{DualAnn}

Let $C$ be a commutative ring. 
A $C$--linear code of length $m$ is, by definition, a $C$--submodule $\mathcal{C}$ of $C^m$. The \emph{dual} of $\mathcal{C}$ is defined as
\[
\mathcal{C}^\bot = \{ w \in C^m : w\transpose{v} = 0, \forall v \in \mathcal{C} \}, 
\]
where $\transpose{M}$ denotes the transpose of a matrix $M$ with coefficients in $C$.

\begin{definition}
A  \emph{Hamming ring extension of rank $m$} is a three-tuple $(C,R,\tovector)$, where $R$ is a ring, $C$ is a commutative subring of $R$ and $\tovector : R \to C^m$ is an isomorphism of $C$--modules. Here, the $C$--module structure of $R$ is given by left multiplication. 
\end{definition}

Hamming ring extensions encode cyclic structures on some $C$--linear codes, according to the following definition. Precise examples will be examined later. 

\begin{definition} 
Every $f \in R$ leads to a $C$--linear code $\mathcal{C} = \tovector(Rf)$. We say then that $\mathcal{C}$ is the \emph{$(C,R,\tovector)$--cyclic code generated by $f$}.  We will also say that $\mathcal{C}$ is \emph{$R$--cyclic}. 
\end{definition}

Given an $R$--cyclic code $\mathcal{C} = \tovector(Rf)$, we have the multiplication map $\cdot f : R \to R$. There is a unique square matrix $M_R(f)  \in \matrixring{m}{C}$ making commute the following diagram of $C$--module morphisms 
\[
\xymatrix{
R \ar[r]^{\cdot f} \ar[d]_{\tovector} & **[r] R \ar[d]^{\tovector} \\
C^m \ar[r]_{\cdot M_R(f)} & **[r] C^m,
}
\]
so \[\mathcal{C} = \operatorname{im}(\cdot M_R(f)).\] 

A straightforward computation shows that 
\begin{equation}\label{dualker}
\mathcal{C}^\bot = \operatorname{ker}(\cdot \transpose{M_R(f)}).
\end{equation}
The map $M_R : R  \to \matrixring{m}{C}$ sending $f$ onto $M_R(f)$ is an injective homomorphism of rings. 

\begin{definition}
Two Haming ring extensions $(C,R,\tovector)$ and $(C,\widehat{R},\widehat{\tovector})$ are said to be \emph{transposed} if there exist an anti-isomorphism of rings $\Theta : R \to \widehat{R}$ such that the diagram
\begin{equation*}\label{eq:TC}
\xymatrix{
R \ar[r]^{M_R} \ar[d]^{\Theta} & **[r] \matrixring{m}{C} \ar[d]^{\transpose{(-)}} \\
\widehat{R} \ar[r]^{M_{\widehat{R}}} & **[r] \matrixring{m}{C}
},
\end{equation*}
is commutative. Equivalently if,
\begin{equation*}
M_{\widehat{R}}(\Theta(f)) = \transpose{M_R(f)}, \ \forall f \in R. 
\end{equation*}
We say also that $\Theta$ is a \emph{transposition} from $(C,R,\tovector)$ to $(C,\widehat{R},\widehat{\tovector})$. 
\end{definition}

Given a subset $X$ of a ring $S$, the \emph{left annihilator} and the \emph{right annihilator} of $X$ are defined, respectively, by 
\[
\lann{S}{X} = \{ s \in S : sx = 0, \forall x \in X \}, \quad \rann{S}{X} = \{ s \in S : xs = 0, \forall x \in X \}.
\]
The first of these sets is a left ideal of $S$, while the second one is a right ideal.

\begin{theorem}\label{Dualcyclic}
Let $(C,R,\tovector)$ and $(C,\widehat{R},\widehat{\tovector})$ be transpose Hamming ring extensions with anti-isomorphism $\Theta : R \to \widehat{R}$. Let \(f,h \in R\) and  \(\mathcal{C} = \tovector(Rf)\). Then \(hR = \rann{R}{Rf}\) if and only if \(\mathcal{C}^\bot = \widehat{\tovector}(\widehat{R}\Theta(h))\). In this case, $\mathcal{C}^\bot$ is generated by the rows of $\transpose{M_R(h)}$.
\end{theorem}

\begin{proof}
Assume $\rann{R}{Rf} = hR$. Then we get $\lann{\widehat{R}}{\Theta(f)\widehat{R}} = \widehat{R}\Theta(h)$. This implies that the top row of the commutative diagram
\[
\xymatrix@C=50pt{\widehat{R} \ar^-{\cdot \Theta(h)}[r] \ar_-{\widehat{\tovector}}[d]& \widehat{R} \ar^-{\cdot \Theta(f)}[r] \ar_-{\widehat{\tovector}}[d] & \widehat{R} \ar_-{\widehat{\tovector}}[d] \\
C^m \ar_-{\cdot M_{\widehat{R}}(\Theta(h))}[r] & C^m \ar_-{\cdot M_{\widehat{R}}(\Theta(f))}[r] & C^m}
\]
is exact, so is the bottom row. Therefore,
\[
\widehat{\tovector}(\widehat{R}\Theta(h)) = \operatorname{im} (\cdot M_{\widehat{R}}(\Theta(h)) = \operatorname{ker}(\cdot M_{\widehat{R}}(\Theta(f)) = \operatorname{ker}(\cdot \transpose{M_R(f)}) = \mathcal{C}^\bot,
\]
where the last equality is \eqref{dualker}. 

Conversely, assume \(\mathcal{C}^\bot = \widehat{\tovector}(\widehat{R}\Theta(h))\). Then
\[
\mathcal{C}^\bot = \operatorname{im}(\cdot M_{\widehat{R}}(\Theta(h))) = \operatorname{im}(\cdot \transpose{M_{R}(h)}),
\]
and therefore \(M_R(f) M_R(h) = 0\). So \(fh = 0\) and \(hR \subseteq \rann{R}{Rf}\). Let \(h' \in \rann{R}{Rf}\), i.e. \(fh' = 0\). Then \(M_R(f) M_R(h') = 0\) and so
\[
\transpose{M_R(h')} \transpose{M_R(f)} = 0.
\]
Since \(\transpose{M_R(h')} = M_{\widehat{R}}(\Theta(h'))\) and \(\mathcal{C}^\bot = \ker(\cdot \transpose{M_R(f)})\), it follows that \(\widehat{\tovector}(\widehat{R} \Theta(h')) = \operatorname{im}(\cdot M_{\widehat{R}}(\Theta(h'))) \subseteq \mathcal{C}^\bot = \widehat{\tovector}(\widehat{R}\Theta(h))\). Hence \(\Theta(h') \in \widehat{R}\Theta(h)\), which implies \(h' \in h R\). Therefore \(\rann{R}{Rf} \subseteq hR\) and the equality holds. 

Finally, the equality $\widehat{\tovector}(\widehat{R}\Theta(h)) = \operatorname{im}(\cdot \transpose{M_R(h)})$, implies that \(\mathcal{C}^\bot\) is generated by the rows of \(\transpose{M_R(h)}\). 
\end{proof}

\begin{corollary}\label{annann}
If $\mathcal{C}^{\bot \bot} = \mathcal{C}$ and $hR = \rann{R}{Rf}$, then $Rf = \lann{R}{hR}$. 
\end{corollary}
\begin{proof}
We have $\mathcal{C}^\bot = \widehat{\tovector}(\widehat{R}\Theta(h))$. Let $\widehat{\Theta} : \widehat{R} \to R$ be the inverse of $\Theta$. Since 
\[
\mathcal{C}^{\bot \bot} = \mathcal{C} = \tovector(Rf) = \tovector(R\widehat{\Theta}\Theta(f)),
\]
we get from Theorem \ref{Dualcyclic}, applied to the transposition $\widehat{\Theta} : \widehat{R} \to R$, that $\rann{\widehat{R}}{\widehat{R}\Theta (h)} = \Theta(f)\widehat{R}$. But this implies that $\lann{R}{hR} = Rf$.
\end{proof}

\begin{remark}\label{pidFrobenius}
It is well known that, if $C$ is a field, then $\mathcal{C}^{\bot \bot} = \mathcal{C}$ for every $C$--linear code $\mathcal{C}$. This equality also holds  true in more general situations of interest. This is the case, for instance, if $C$ is a principal ideal domain and $\mathcal{C}$ is a direct summand of $C^m$, or $\mathcal{C}$ is any $C$--linear code over a Frobenius ring $C$. 
\end{remark}

\section{Dual left ideal convolutional codes}\label{IdealCodes}

The method described in Section \ref{DualAnn} is abstracted from the study of the dual of a group convolutional code developed in \cite[Section 4]{LopezPermouth/Szabo:2013}. The ideas from \cite[Section 4]{LopezPermouth/Szabo:2013} were adapted in \cite{GLN:2016} to $\field$--linear cyclic convolutional codes when the word ambient algebra is a matrix algebra $\matrixring{n}{\field}$. The aim of this section is to extend the results on duality from \cite{GLN:2016} to the matrix $\field$--algebra $\matrixring{n}{\mathbb{K}}$, where $\mathbb{K}$ is a finite field extension of $\field$. To this end, we first work, in the spirit of Section \ref{DualAnn}, in a more general setting, and then apply the general results to the more concrete situation. 

\subsection{Left ideal convolutional codes of automorphism type.}
Let $A$ be a ring and $\sigma : A \to A$ a ring automorphism. The skew \emph{right} polynomial ring $\rOre{A}{z}{\sigma}$ is defined as the free right $A$--module with basis $\{ z^i : i \in \mathbb{N} \}$ with the multiplication determined by the rules $z^iz^j = z^{i+j}$, and $az = z\sigma(a)$, for all $i, j \in \mathbb{N}$ and $a \in A$. If $A$ is a finite algebra over a finite field $\field$, and $\sigma$ is an $\field$--automorphism, then $R = \rOre{A}{z}{\sigma}$ becomes the sentence ambient algebra for some cyclic convolutional codes as follows. First, observe that $\field{}[z]$, the commutative polynomial ring in the variable $z$, is a subring of $R$. Moreover, each $\field$--basis $B = \{ v_0, v_1, \dots, v_{m-1} \}$ of $A$ leads to the associated coordinate map $\tovector : A \to \field^m$. Since $B$ becomes a basis of $R$ as a (left) $\field{}[z]$--module, we get that $\tovector$ extends to an $\field{}[z]$--module isomorphism $\tovector : R \to \field{}[z]^m$. That is, $(\field{}[z], R, \tovector)$ becomes a Hamming ring extension of rank $m$. 

Convolutional codes may be understood as $\field{}[z]$--submodules of $\field{}[z]^m$, so, they may be considered as $\field{}[z]$--linear codes. The variable $z$ is interpreted as the delay operator \cite{Forney:1970}.  Convolutional codes  are often required to be, in addition, direct summands of $\field{}[z]^m$. 

\begin{definition}\cite{LopezPermouth/Szabo:2013}\label{licc}
A direct summand $\field{}[z]$--submodule $\mathcal{C}$ of $\field{}[z]^m$ is said to be a \emph{left ideal convolutional code} if there exits a left ideal $I$ of $\rOre{A}{z}{\sigma}$ such that $\mathcal{C} = \tovector(I)$.  
\end{definition}

Our next aim is to prove that, under suitable conditions, it is possible to construct  a transpose Hamming ring extension to $R$. For $a \in A$, we use the notation $M_a = M_R(a)$. Observe that $M_a \in \matrixring{m}{\field}$. On the other hand, for every $\field$--linear map $\lambda : A \to A$,   let $M_{\lambda} \in \matrixring{m}{\field}$ be the unique matrix such that the diagram
\[
\xymatrix{A \ar^{\lambda}[r] \ar^{\tovector}[d] &A  \ar^{\tovector}[d] \\
\field^m \ar^{\cdot M_{\lambda}}[r]  & \field^m}
\]
is commutative.  A straightforward computation shows that, if $\lambda$ is an algebra map, then, for all $a \in A$, 
\begin{equation}\label{eq:alambda}
M_a M_\lambda = M_\lambda M_{\lambda(a)}
\end{equation}

\begin{proposition}\label{descriptionofM_Rsigma}
For every \(f = \sum_kz^kf_k \in R\) we have:
\[M_R(f) = \sum_k z^k M_{\sigma}^k M_{f_k}.\]
\end{proposition}

\begin{proof}
This proof is adapted and simplified from \cite[Proposition 4.7]{LopezPermouth/Szabo:2013} and \cite[Proposition 2.4]{GLN:2016}.
Since \(\tovector : R \to \field{}[z]^m\) is an \(\field{}[z] \)--linear map, we have, for every $g \in R$, 
\[
\begin{split}
\tovector(g) \textstyle\sum_k z^k M_{\sigma}^k M_{f_k} &= \tovector\left(\textstyle\sum_l z^l g_l\right) \textstyle\sum_k z^k M_{\sigma}^k M_{f_k} \\
&= \textstyle\sum_{l,k} z^l z^k \tovector(g_l) M_{\sigma^k} M_{f_k} \\
&= \textstyle\sum_{l,k} z^l z^k \tovector(\sigma^k(g_l)f_k) \\
&= \tovector\left( \textstyle\sum_l z^l \textstyle\sum_k z^k \sigma^k(g_l)f_k \right) \\
&= \tovector\left( \textstyle\sum_l z^l \textstyle\sum_k g_l z^k f_k \right) \\
&= \tovector(gf).
\end{split}
\]
\end{proof}

Consider now an involution $\theta : A \to A$ (that is, an anti-algebra homomorphism such that $\theta^2 = \identity{A}$). Let us denote $\widehat{\sigma} = \theta \sigma^{-1} \theta$. Define the map 
\begin{equation}\label{antimapOre}
\begin{split}
\Theta : \rOre{A}{z}{\sigma} &\to \rOre{A}{z}{\widehat{\sigma}} \\
\textstyle\sum_k z^k a_k &\mapsto \Theta(\textstyle\sum_k z^k a_k) = \sum_k z^k \theta \sigma^{-k}(a_k)
\end{split}
\end{equation}
Note that, as left $\field{}[z]$--modules, both rings are equal, and, once the basis $B$ of $A$ is fixed, we have that $(\field{}[z], \rOre{A}{z}{\widehat{\sigma}}, \tovector)$ is a Hamming ring extension.

\begin{proposition}\label{FzTranspose}
The map \(\Theta : \rOre{A}{z}{\sigma} \to \rOre{A}{z}{\widehat{\sigma}} \) is an anti-isomorphism of $\field$--algebras. Moreover, if  
\(M_{\widehat{\sigma}} = \transpose{M_\sigma}\) and \(M_{\theta(a)} = \transpose{M_a}\) for all \(a \in A\), then $(\field{}[z], \rOre{A}{z}{\sigma}, \tovector)$ and $(\field{}[z], \rOre{A}{z}{\widehat{\sigma}}, \tovector)$ are transpose Hamming ring extensions via $\Theta$.
\end{proposition}

\begin{proof}
Since \(\theta\) and \(\sigma\) are \(\field{}\)--linear and bijective, \(\Theta\) is also \(\field{}\)--linear and bijective. Let us check that $\Theta$ is anti-multiplicative. Obviously, the restriction of $\Theta$ to $A$ is just $\theta$. Since $\rOre{A}{z}{\sigma}$ is generated as an $\field$--algebra by $A$ and $z$, it is enough to consider $a \in A$ and compute
\[
\Theta(a)\Theta(z) = \theta(a) z = z \theta \sigma^{-1} \theta \theta(a) = z \theta \sigma^{-1}(a) = \Theta(za)
\]
and
\[
\Theta(az) = \Theta(z \sigma(a)) = z \theta\sigma^{-1}\sigma(a) = z \theta(a) = \Theta(z) \Theta(a).
\]
In order to check that $\Theta$ is a transposition, let us compute
\[
\begin{split}
\transpose{M_R(f)}
&= \transpose{\left[ \sum_k z^k M_{\sigma}^k M_{f_k} \right]} \quad \text{by Proposition \ref{descriptionofM_Rsigma}}\\
&= \transpose{\left[ \sum_k z^k M_{\sigma^{-k}(f_k)} M_{\sigma}^k \right]} \quad \text{by \eqref{eq:alambda}}\\
&= \sum_k z^k \transpose{(M_{\sigma}^k)} \transpose{(M_{\sigma^{-k}(f_k)})} \\
&= \sum_k z^k (\transpose{M_{\sigma}})^k M_{\theta{\sigma^{-k}(f_k)}} \quad \text{by assumption}\\
&= \sum_k z^k (M_{\widehat{\sigma}})^k M_{\theta{\sigma^{-k}(f_k)}} \quad \text{by assumption}\\
&= M_{\widehat{R}}(\Theta({f})) \quad \text{by Proposition \ref{descriptionofM_Rsigma},}
\end{split}
\]
as desired. 
\end{proof}

\begin{remark}
If $A = \field G$ is the group algebra of a finite group $G$, and $\theta$ is the involution defined on $G$ as $\theta (g) = g^{-1}$, then Proposition \ref{FzTranspose} gives \cite[Proposition 4.18]{LopezPermouth/Szabo:2013} in the case where the skew derivation used in \cite{LopezPermouth/Szabo:2013} is zero. 
\end{remark}

\begin{theorem}\label{dualidealcode}
Assume that  
\(M_{\widehat{\sigma}} = \transpose{M_\sigma}\) and \(M_{\theta(a)} = \transpose{M_a}\) for all \(a \in A\). Let $f \in R = \rOre{A}{z}{\sigma}$ such that $\rann{R}{Rf} = hR$ for some $h \in R$. If $\mathcal{C} = \tovector (Rf)$,  then $\mathcal{C}^\bot = \tovector(\widehat{R}\Theta(h))$, and both $\mathcal{C}$ and $\mathcal{C}^\bot$ are left ideal convolutional codes. Moreover, $Rf = \lann{R}{hR}$.
\end{theorem}

\begin{proof}
The equality $\mathcal{C}^\bot = \tovector(\widehat{R}\Theta(h))$ follows from Proposition \ref{Dualcyclic} and Proposition \ref{FzTranspose}. As observed in the proof of Theorem \ref{Dualcyclic}, $\widehat{R}\Theta(h) = \lann{\widehat{R}}{\Theta(f)}$. Then $\mathcal{C}^\bot$ is a left ideal convolutional code by \cite[Lemma 2.1]{GLN:2016}. Corollary \ref{annann}, by virtue of Remark \ref{pidFrobenius}, gives the equality $Rf = \lann{R}{hR}$. Therefore, $\mathcal{C}$ is a left ideal convolutional code, by \cite[Lemma 2.1]{GLN:2016}.
\end{proof}

\begin{remark}
If $\field{}[z] \subseteq R$ is a separable ring extension, then every left ideal convolutional code is of the form $\tovector(Re)$ for some idempotent $e \in R$ (see \cite{GLN2017IdealCodes}). Therefore, $\rann{R}{Re} = (1-e)R$. 
\end{remark}

\begin{remark}
If $A = \field G$, the group algebra of a finite group $G$, then Theorem \ref{dualidealcode} gives, in this particular case, a stronger statement than \cite[Theorem 4.21]{LopezPermouth/Szabo:2013}, for the case where the $\sigma$--derivation considered there is zero. 
\end{remark}

\subsection{Convolutional codes with a simple word-ambient algebra.}

In \cite{GLN:2016}, dual codes of left ideal convolutional codes, when \(A\) is the matrix algebra \(\matrixring{n}{\field{}}\), are studied. Concretely, \cite[propositions 2.9 and 2.10]{GLN:2016} provide a transposition of Hamming ring extensions in this special case, and, henceforth, Theorem \ref{dualidealcode} generalizes and improves \cite[Theorem 2.12]{GLN:2016}. 

Our next aim is to give sufficient conditions to apply Proposition \ref{FzTranspose} and Theorem \ref{dualidealcode} to \(A = \matrixring{n}{\mathbb{K}}\),  where \(\mathbb{K} = \field[q^t]\) is a finite field extension of \(\field{} = \field[q] \) of degree \(t\).  Let \(D = \{ \alpha_0, \dots, \alpha_{t-1} \}\) be a basis of \(\mathbb{K}\) as an \(\field{}\)--vector space. We have a monomorphism of $\field$--algebras
\[
\begin{split}
\tomatrix : \mathbb{K} &\to \matrixring{t}{\field{}} \\
\gamma &\mapsto \tomatrix(\gamma),
\end{split}
\]
where $\tomatrix(\gamma)$ is the matrix that represents the multiplication map $\cdot \gamma : \mathbb{K} \to \mathbb{K}$ for $\gamma \in \mathbb{K}$ with respect to $D$.

Let $V, V'$ be finite-dimensional $\mathbb{K}$--vector spaces with bases $B = \{ v_0, \dots, v_{r-1}\}$ and $B' = \{v'_0, \dots, v'_{s-1} \}$, respectively. Define the $\mathbb{F}$--basis \[ B_{\field} = \{ \alpha_iv_j : 0 \leq i \leq t, 0 \leq j \leq r \}\] of $V$ ordered by the condition that $\alpha_i v_j$ is before than $\alpha_k v_{j+1}$ for all $i,j,k$. The $\field$--basis $B'_{\field}$ of $V'$ is defined analogously. 

\begin{lemma}\label{matrix_expansion}\label{tomatrix_multiplicative}
Let \(M = (m_{ij}) \in \matrixring{r \times s}{\mathbb{K}}\) be the matrix associated with respect to the bases $B$ and $B'$ to a $\mathbb{K}$--linear map $\lambda : V \to V'$.  The matrix associated to $\lambda$ considered as an $\field$--linear map with respect to the bases $B_{\field}$ and $B'_{\field}$  is
\[
\tomatrix(M) = \Big( \tomatrix(m_{ij}) \Big)_{0 \leq i < r \atop 0 \leq j < s} \in \matrixring{rt \times st}{\field{}}.
\]
\end{lemma}

\begin{proof}
Straightforward.
\end{proof}

Consider the $\mathbb{K}$--basis  \(\mathcal{B} = \{ E_{ij} ~|~ 0 \leq i,j < n, 0 \leq k < t \}\) of $A = \matrixring{n}{\mathbb{K}}$, where \(E_{ij}\) is the matrix with \(1\) in the position corresponding to \(i\)th row and \(j\)th column, and \(0\) elsewhere. Order $\mathcal{B}$ in such a way that the corresponding coordinate map $\matrixring{n}{\mathbb{K}} \to \mathbb{K}^{n^2}$  writes each matrix as the concatenation of its rows. We thus get the $\field$--basis $\mathcal{B}_{\field} = \{ \alpha_kE_{ij}  : 0 \leq k \leq t-1, 0 \leq i, j \leq n-1 \}$ ordered as described above. This last basis leads to a coordinate map $\tovector : A \to \mathbb{F}^{tn^2}$ which, for any $\field$--automorphism $\sigma$ of $A$, extends to an $\mathbb{F}[z]$--linear isomorphism $\tovector : \rOre{A}{z}{\sigma} \to \mathbb{F}[z]^{tn^2}$. Our aim is to prove that, if the basis $D$ of $\mathbb{K}$ over $\mathbb{F}$ is a normal self-dual basis, then a transposition $\Theta$ fulfilling the hypotheses of Proposition \ref{FzTranspose} can be constructed.

Let  \(D^* = \{ \beta_0, \dots, \beta_{t-1} \}\) the dual basis of $D = \{\alpha_0, \dots, \alpha_{t-1} \}$. Recall that \(D\) is normal if there exists \(\alpha \in \mathbb{K}\) such that \(\alpha_k = \alpha^{q^k}\) for all \(0 \leq k < t\). Also recall that \(D \) is self-dual if \(\alpha_k = \beta_k\) for all \(0 \leq k < t\).

The coordinate map with respect to $D$ is given by
\[
 \mathbb{K} \to \field{}^t, \quad
\gamma \mapsto \left(\trace(\beta_0 \gamma), \dots, \trace(\beta_{t-1} \gamma) \right) 
\]
where \(\trace \) denotes the trace map of the field extension \(\field{} \subseteq \mathbb{K}\). 
By the properties of the trace function,
\[
\tomatrix(\gamma) =  \Big( \trace(\beta_i \gamma \alpha_j) \Big)_{0 \leq i,j < t}.
\]

As involution on \(A\), we use the natural one, that is, \(\theta : A \to A\) be the involution given by \(\theta(a) = \transpose{a}\).  In order to show how $\theta$ fulfills the Proposition \ref{FzTranspose}, we need to use some properties of the Kronecker product of matrices over a field $L$.

Recall that if \(M \in \matrixring{r \times s}{L}\) and \(N \in \matrixring{r' \times s'}{L}\), the Kronecker product of \(M\) and \(N\) is defined as 
\[
M \kronecker N = \left( \begin{matrix}
m_{0,0} N & m_{0,1} N & \cdots & m_{0,s-1} N \\
m_{1,0} N & m_{1,1} N & \cdots & m_{1,s-1} N \\
\vdots & \vdots & \ddots & \vdots \\
m_{r-1,0} N & m_{r-1,1} N & \cdots & m_{r-1,s-1} N 
\end{matrix}\right),
\]
where \(M = \big( m_{i,j} \big)_{0 \leq i < r, 0 \leq j < s}\). Properties of the Kronecker product can be seen in \cite[Chapter 4]{Horn/Johnson:1994}.

\begin{lemma}\label{propertiesofMa}
The following properties hold for all \(a \in A\):
\begin{enumerate}[(i)]
\item \(M_a = \tomatrix(I \kronecker a)\).
\item If \(B\) is self-dual, \(\transpose{M_a} = M_{\transpose{a}}\).
\end{enumerate}
\end{lemma}

\begin{proof}
(i) follows from Lemma \ref{matrix_expansion} and \cite[Lemma 2.2]{GLN:2016}. If \(D\) is self-dual then \(\tomatrix(\gamma)\) is symmetric for all \(\gamma \in \mathbb{K}\), hence (ii) follows from (i) and \cite[Lemma 2.3(iii)]{GLN:2016}. 
\end{proof}

Lemma \ref{propertiesofMa} ensures the first hypothesis of Proposition \ref{FzTranspose}. As for the second one concerns, namely the equality \(M_{\widehat{\sigma}} = \transpose{M_\sigma}\), it will be obtained with the help of a suitable decomposition of $\sigma$. Let \(\tau : \mathbb{K} \to \mathbb{K}\) denote the Frobenius $\mathbb{F}$--automorphism of $\mathbb{K}$, i.e. \(\tau(\gamma) = \gamma^q\).  By \cite[Theorem 2.4]{Cauchon/Robson:1978}, for every $\mathbb{F}$--automorphism $\sigma$ of the matrix algebra $A$,  there exist a regular matrix \(U \in A\) and \(0 \leq h \leq t-1\) such that \(\sigma = \sigma_U \circ \sigma_{\tau^h}\), where \(\sigma_U\) is the inner automorphism associated to \(U\), i.e. \(\sigma_U(a) = U a U^{-1}\), and \(\sigma_{\tau^h}\) is the componentwise extension of \(\tau^h\) to \(A\), i.e. \(\sigma_{\tau^h}\big( m_{ij} \big) = \big( \tau^h(m_{ij}) \big)\).

We are going to analyze inner automorphisms and extensions of field automorphisms independently, and later we will join both results. 

For each regular matrix \(U \in A\), a straightforward computation shows that \(\widehat{\sigma_U} = \theta \sigma_U^{-1} \theta = \sigma_{\transpose{U}}\).

\begin{lemma}\label{propertiesofMsigmaU}
The following properties hold:
\begin{enumerate}[(i)]
\item \(M_{\sigma_U} = \tomatrix(\transpose{U} \kronecker U^{-1})\).
\item If \(D\) is self-dual, then \(M_{\widehat{\sigma_U}} = \transpose{M_{\sigma_U}}\).
\end{enumerate}
\end{lemma}

\begin{proof}
(i) follows from Lemma \ref{matrix_expansion} and \cite[Lemma 2.3(ii)]{GLN:2016}. If \(D\) is self-dual then \(\tomatrix(\gamma)\) is symmetric for all \(\gamma \in \mathbb{K}\), hence (ii) follows from (i).
\end{proof}

\begin{proposition}\label{TCsigmaU}
Let \(R = \rOre{A}{z}{\sigma_U}\) and \(\widehat{R} = \rOre{A}{z}{\widehat{\sigma_U}}\). If \(D\) is a self-dual basis,  then \((\field{}[z],R,\tovector)\) and \((\field{}[z],\widehat{R},\tovector)\) are transposed Hamming extensions. 
\end{proposition}

\begin{proof}
By Lemma \ref{propertiesofMa}, \(M_{\transpose{a}} = \transpose{M_a}\) for all \(a \in A\). By Lemma \ref{propertiesofMsigmaU}, \(M_{\widehat{\sigma_U}} = \transpose{M_{\sigma_U}}\). Hence the result follows from Proposition \ref{FzTranspose}.
\end{proof}

Now we focus on extensions of field automorphisms. 

\begin{lemma}\label{Ph}
Let \(\tau\) be the Frobenius automorphism of $\mathbb{K}$ over $\mathbb{F}$, and let \(0 \leq h < t\). If \(D\) is a normal basis, then the automorphism \(\tau^h\) is represented as \(\field{}\)--linear map by the matrix 
\[
P_h = \begin{pmatrix} 0 & I_{t-h} \\ I_h & 0 \end{pmatrix} \in \matrixring{t}{\field{}}.
\]
The automorphism \(\sigma_{\tau^h} : A \to A\) is represented as \(\field{}\)--linear map by the matrix 
\[
M_{\sigma_{\tau^h}} = I_{n^2} \kronecker P_h = \begin{pmatrix} P_h & & \\ & \ddots & \\ & & P_h \end{pmatrix} \in \matrixring{n^2t}{\field{}}. 
\]
\end{lemma}

\begin{proof}
Straightforward since \(\tau\) performs a cyclic permutation of one position to the right of the elements of \(D\). 
\end{proof}

\begin{lemma}\label{propertiesofMh}
Assume \(D\) to be a normal basis. The following properties hold:
\begin{enumerate}[(i)]
\item \(M_{\sigma_{\tau^h}}^{-1} = M_{\sigma_{\tau^{t-h}}} = \transpose{M_{\sigma_{\tau^h}}}\).
\item \(M_{\widehat{\sigma_{\tau^h}}} = \transpose{M_{\sigma_{\tau^h}}}\),
\end{enumerate}
\end{lemma}

\begin{proof}
(i) is a direct consequence of Lemma \ref{Ph}. Since \(\sigma_{\tau^h}\) acts component-wise, it commutes with \(\theta\), hence \(\widehat{\sigma_{\tau^h}} = \theta \sigma_{\tau^h}^{-1} \theta = \theta^2 \sigma_{\tau^h}^{-1} = \sigma_{\tau^h}^{-1} = \sigma_{\tau^{t-h}}\), hence (ii) follows from (i).
\end{proof}

\begin{proposition}\label{TCsigmah}
Let \(R = \rOre{A}{z}{\sigma_{\tau^h}}\) and \(\widehat{R} = \rOre{A}{z}{\widehat{\sigma_{\tau^h}}}\). If \(D\) is a normal basis then \((\field{}[z],R,\tovector)\) and \((\field{}[z],\widehat{R},\tovector)\) are transposed Hamming extensions.
\end{proposition}

\begin{proof}
Again, Lemma \ref{propertiesofMa} implies \(M_{\transpose{a}} = \transpose{M_a}\) for all \(a \in A\). Lemma \ref{propertiesofMh} implies \(M_{\widehat{\sigma_{\tau^h}}} = \transpose{M_{\sigma_{\tau^h}}}\). So Proposition \ref{FzTranspose} gives the result. 
\end{proof}

Finally we get to general automorphisms \(\sigma = \sigma_U \sigma_{\tau^h} \in \Aut[\field{}]{A}\) where \(U \in A\) is a regular matrix and \(0 \leq h < t\). 

\begin{lemma}\label{propertiesofMsigma}
Let \(\sigma = \sigma_U \sigma_{\tau^h} \in \Aut[\field{}]{A}\). If \(D\) is a self-dual normal basis, we have
\[
M_{\widehat{\sigma}} = \transpose{M_{\sigma}}
\]
\end{lemma}

\begin{proof}
The decomposition \(\sigma = \sigma_U \sigma_{\tau^h}\) implies \(M_\sigma = M_{\sigma_{\tau^h}} M_{\sigma_U}\). Since
\[
\widehat{\sigma} = \widehat{\sigma_U \sigma_{\tau^h}} = \theta \sigma_{\tau^h}^{-1} \sigma_U^{-1} \theta = \theta \sigma_{\tau^h}^{-1} \theta \theta \sigma_U^{-1} \theta = \widehat{\sigma_{\tau^h}} \widehat{\sigma_U},
\]
we get \(M_{\widehat{\sigma}} = M_{\widehat{\sigma_U}} M_{\widehat{\sigma_{\tau^h}}}\). Hence the result follows from Lemma \ref{propertiesofMsigmaU} and Lemma \ref{propertiesofMh}.
\end{proof}

\begin{proposition}\label{TCsigma}
Let \(R = \rOre{A}{z}{\sigma}\) and \(\widehat{R} = \rOre{A}{z}{\widehat{\sigma}}\). If \(D\) is a self-dual normal basis, then \((\field{}[z],R,\tovector)\) and \((\field{}[z],\widehat{R},\tovector)\) are transposed Hamming extensions.
\end{proposition}

\begin{proof}
Decompose $\sigma =  \sigma_U \sigma_{\tau^h}$ according to \cite[Theorem 2.4]{Cauchon/Robson:1978}.  The proof is now completely analogous to Proposition \ref{TCsigmaU} and Proposition \ref{TCsigmah} by using Lemma \ref{propertiesofMa}, Lemma \ref{propertiesofMsigma} and Proposition \ref{FzTranspose}.
\end{proof}

We are now in position to state the main result of this subsection, which is a consequence of Theorem \ref{dualidealcode} and Proposition \ref{TCsigma}. Field extensions which have a self-dual normal basis are characterized by results of Lempel, Weinberger, Seroussi, Imamura and Morii, see \cite{Lempel/Weinberger:1988} and its references. 

\begin{theorem}\label{matrixmain}
Let $\sigma$ be any $\mathbb{F}$--automorphism of  \(A = \matrixring{n}{\mathbb{K}}\).  Assume there exists a self-dual normal basis $D$ of $\mathbb{K}$ over $\mathbb{F}$.  Let \(f \in R = \rOre{A}{z}{\sigma}\) such that \(\rann{R}{Rf} = hR\) for some $h \in R$. If $\mathcal{C} = \tovector (Rf)$, then $\mathcal{C}^\bot = \tovector(\widehat{R}\Theta(h))$, and $\mathcal{C}$ and $\mathcal{C}^\bot$ are both left ideal convolutional codes. Moreover, $Rf = \lann{R}{hR}$.
\end{theorem}
\begin{proof}
The result follows from Proposition \ref{TCsigma} and Theorem  \ref{dualidealcode}. 
\end{proof}

\begin{remark}
Taking $\mathbb{K} = \mathbb{F}$ in Theorem \ref{matrixmain} leads to a strong form of \cite[Theorem 2.12]{GLN:2016}.
\end{remark}

\begin{remark}
By \cite[Theorems 1 and 2]{Lempel/Weinberger:1988}, if \(t\) is odd, or \(q\) is even and \(t \equiv 2 \mod{4}\), then there exists a self-dual normal basis \(D\) of \(\mathbb{K}\) over \(\field{}\), and Theorem \ref{matrixmain} applies. 
\end{remark}

\section{Dual of skew constacyclic codes}\label{sec:constacyclic}

Let $C$ be a commutative ring and $\sigma$ an automorphism of $C$. Assume that $\sigma$ has finite order, and let $n \geq 1$ such that $\sigma^n = \identity{C}$. If $u \in C$ is a unit such that $\sigma(u) = u$, then the both $x^n-u$ and $x^n -u^{-1}$ are central elements in the skew left polynomial ring $\lOre{C}{x}{\sigma}$, where the multiplication rule is now \(x a = \sigma(a) x\).  Consider the factor rings
\[
\mathcal{R} = \frac{\lOre{C}{x}{\sigma}}{\langle x^n - u \rangle}, \qquad \widehat{\mathcal{R}} = \frac{\lOre{C}{x}{\sigma}}{\langle x^n - u^{-1} \rangle},
\]
which contain $C$ as a subring. 
Since the ideal $\langle x^n-u \rangle$ coincides with the left ideal generated by $x^n - u$, and this polynomial is monic, we easily get that $\{1, x, \dots, x^{n-1} \}$ is a basis of $\mathcal{R}$ as a left $C$--module (we are identifying, as usual, each equivalent class modulo $x^n -u$ 
with its unique representative of degree less than $n$). We have a Hamming ring extension $(C,\mathcal{R},\tovector)$, where $\tovector : \mathcal{R} \to C^n$ is the coordinate map with respect to the aforementioned basis. 

Since $x$ is a unit in both rings, we may define 
\[
\Theta : \mathcal{R} \to \widehat{\mathcal{R}}, \qquad \sum_{i = 0}^{n-1}a_ix^i \mapsto \sum_{i = 0}^{n-1}\sigma^{-i}(a_i)x^{-i}. 
\]

\begin{lemma}
The map $\Theta$ is an anti-isomorphim of rings. 
\end{lemma}
\begin{proof}
Since $\Theta$ is clearly additive, in order to prove that it is an anti-homomorphism of rings, it suffices to check that it is anti-multiplicative on monomials. To this end, let us first observe that, since $x^n = u^{-1}$ in $\widehat{\mathcal{R}}$, we have for all integers $i,j,k$ with $k = in + j$
\begin{equation}\label{xk}
x^{-k} = u^ix^{-j}
\end{equation}
Now, given any integer $k$, write $k = in + j$ for some integers $i, j$ with $0 \leq j < n$. Then, for any $a \in C$, we have
\[
\Theta(ax^k)
= \Theta (au^ix^j) = u^i\sigma^{-j}(a)x^{-j} = \sigma^{-k}(a)x^{-k},
\]
where, in the last equality, we used \eqref{xk} and that $\sigma^n = \identity{C}$. Therefore, for any $a, b \in C$ and integers $i,j$ we have
\[
\Theta(ax^ibx^j) = \Theta(a\sigma^i(b)x^{i+j}) = \sigma^{-i-j}(a\sigma^i(b))x^{-i-j} = \sigma^{-i-j}(a)\sigma^{-j}(b)x^{-i-j}, 
\]
while
\[
\Theta(bx^j)\Theta(ax^i) = \sigma^{-j}(b)x^{-j}\sigma^{-i}(a)x^{-i} = \sigma^{-j}(b)\sigma^{-j-i}(a)x^{-i-j}. 
\]
Therefore, $\Theta$ is anti-multiplicative. The map $\Theta$ is proved to be bijective by defining a map $\widehat{\Theta} : \widehat{\mathcal{R}} \to \mathcal{R}$ analogous to $\Theta$ which turns out to be its inverse. 
\end{proof}

The set $\{1, x, \dots, x^{n-1} \}$  provides bases as left $C$--modules of $\mathcal{R} = \lOre{C}{x}{\sigma}/\langle x^n - u\rangle$ and of $\widehat{\mathcal{R}} = \lOre{C}{x}{\sigma}/\langle x^n - u^{-1}\rangle$.

\begin{proposition}\label{constatranspose}
Consider the basis $\{1, x, \dots, x^{n-1} \}$ of $\mathcal{R}$ and $\widehat{\mathcal{R}}$, with corresponding coordinate isomorphisms $\tovector : \mathcal{R} \to C^n$ and $\widehat{\tovector} : \widehat{\mathcal{R}} \to C^n$. The anti-isomorphism $\Theta : \mathcal{R} \to \widehat{\mathcal{R}}$ is a transposition from $(C,\mathcal{R},\tovector)$ to $(C,\widehat{\mathcal{R}},\widehat{\tovector})$. 
\end{proposition}
\begin{proof}
First, observe that if \(f = \sum_{i=0}^{n-1} a_i x^i \in \mathcal{R}\) then
\[
M_{\mathcal{R}}(f) = 
\begin{pmatrix}
a_0 & a_1 & \dots & a_{n-1} \\
u \sigma(a_{n-1}) & \sigma(a_0) & \dots & \sigma(a_{n-2}) \\
\vdots & \vdots & \ddots & \vdots \\
u \sigma^{n-1}(a_1) & u \sigma^{n-1}(a_2) & \dots & \sigma^{n-1}(a_0)
\end{pmatrix}.
\]
Now, since $x^{-1} = ux^{n-1} \in \widehat{\mathcal{R}}$, we get that $x^{-i} = ux^{n-i}$ for $i = 0, 1, \dots, n-1$. Hence,
\[
\Theta (f) = \sum_{i=0}^{n-1}\sigma^{n-i}(a_i)ux^{n-i} = a_0 + \sum_{j=1}^{n-1} u \sigma^j(a_{n-j}) x^j.
\]
Therefore,
\[
M_{\widehat{\mathcal{R}}}(\Theta(f)) = 
\begin{pmatrix}
a_0 & u\sigma(a_{n-1}) & \dots & u\sigma^{n-1}(a_1) \\
a_1 & \sigma(a_0) & \dots & u\sigma^{n-1}(a_2) \\
\vdots & \vdots & \ddots & \vdots \\
a_{n-1} & \sigma(a_{n-2}) & \dots & \sigma^{n-1}(a_0)
\end{pmatrix} = \transpose{M_{\mathcal{R}}(f)}
\]
because $\sigma^n = \identity{C}$. 
\end{proof}

\begin{definition}
Let \((C,\mathcal{R},\tovector)\) be the Hamming ring extension where \(\mathcal{R} = \frac{\lOre{C}{x}{\sigma}}{\langle x^n - u \rangle}\) with \(\sigma^n = \identity{C}\) and \(u \in C\) is a unit such that \(\sigma(u) = u\). A \((C,\mathcal{R},\tovector)\)-cyclic code is called \((u,\sigma)\)-constacyclic code. 
\end{definition}

\begin{theorem}\label{dualconstacyclic}
Let $f \in \lOre{C}{x}{\sigma}$ be monic such that $x^n - u = fh$ for some $h \in \lOre{C}{x}{\sigma}$. Let $\mathcal{C} = \tovector(\mathcal{R}f)$ be the $(u,\sigma)$--constacyclic $C$--linear code generated by $f$. Then $\mathcal{C}^\bot = \widehat{\tovector}(\widehat{\mathcal{R}}\Theta(h))$, and it is thus a $(u^{-1},\sigma)$--constacyclic code.
\end{theorem}

\begin{proof}
Let us first check that $\rann{\mathcal{R}}{\mathcal{R} f} = h\mathcal{R}$. By hypotesis, $h \in \rann{\mathcal{R}}{\mathcal{R} f}$ so that $\rann{\mathcal{R}}{\mathcal{R} f}\supseteq h\mathcal{R}$. To see the converse inclusion, let $h' \in \rann{\mathcal{R}}{\mathcal{R} f}$. Then $fh' = (x^n-u)g = fhg$ for some $g \in \lOre{C}{x}{\sigma}$. Since monic polynomials are nonzero divisors in $\lOre{C}{x}{\sigma}$, we get that $h' = hg$ as desired.  

By virtue of Proposition \ref{constatranspose} we may apply Theorem \ref{Dualcyclic} and thus get $\mathcal{C}^\bot = \widehat{\tovector}(\widehat{\mathcal{R}}\Theta (h))$.
\end{proof}

\begin{remark}\label{u2=1}
If \(u^2 = 1\), we get \(\mathcal{R} = \widehat{\mathcal{R}}\) and \(\Theta\) becomes an involution in \(\mathcal{R}\), i.e. an anti-algebra automorphism such that \(\Theta^2 = \identity{\mathcal{R}}\). 
\end{remark}
  
\begin{corollary}\label{dualu2=1}
If $u^2 = 1$, then $\mathcal{C}^\bot = \tovector(\mathcal{R}\Theta(h))$. 
\end{corollary}

\begin{example}
Cyclic and negacyclic codes over finite chain rings defined in \cite{Dinh/LopezPermouth:2004} fit in our construction. Let \(C\) be a commutative chain ring and \(\mathcal{C}\) a \(C\)-linear code of length \(n\). Then \(\mathcal{C}\) is cyclic if and only if \(\tovector^{-1}(\mathcal{C})\) is an ideal of \(\mathcal{R} = C[x]/\langle x^n-1 \rangle = \widehat{\mathcal{R}}\). The ring \(\mathcal{R}\) is proven to be a principal ideal ring, \cite[Corollary 3.7]{Dinh/LopezPermouth:2004}, and generators \(F\) and \(G\) for a cyclic code \(\mathcal{C}\) and its dual \(\mathcal{C}^\bot\) are computed, \cite[theorems 3.6 and 3.10]{Dinh/LopezPermouth:2004}. Theorem \ref{Dualcyclic} implies that \(\Theta(G)\mathcal{R} = \rann{\mathcal{R}}{\mathcal{R}F}\). The same can be said for negacyclic \(C\)-codes, i.e. ideals of \(C[x]/\langle x^n + 1 \rangle\), if \(F\) and \(G\) are the generators computed in \cite[Theorem 5.7 and 5.12]{Dinh/LopezPermouth:2004} then \(\Theta(G)\) generates the annihilator of \(F\). 
\end{example}

\begin{example}
Since, in Theorem \ref{dualconstacyclic},  $x^n - u$ is a central element of $\lOre{C}{x}{\sigma}$, the equality $fh = x^n-u$ implies, by a standard argument, that $hf = x^n -u$. Now, $x$ is a unit in $\widehat{\mathcal{R}}$, which implies that  $x^k\Theta(h)$ generates $\widehat{\mathcal{R}} \Theta(h)$ for every integer $k$. Setting $k$ equal to the degree of $h$, we get that $x^k\Theta(h)$ is a generator polynomial of the dual code $\mathcal{C}^\bot$ which takes the form of the given in \cite[Corollary 18]{Boucher/Ulmer:2009} for $\sigma$--codes over a finite field. We also get the generator of the dual of any skew constacyclic code computed in \cite[Theorem 4.4]{Boucher/alt:2008}, \cite[Lemma 3.1]{Jitman/alt:2012} and  \cite[Proposition 3]{Ducoat/Oggier:2016} in the realm of codes over finite commutative rings. 
\end{example}

\section{Dual of skew Reed-Solomon codes}\label{DualsRS}

Let $\sigma$ be an automorphism of finite order $n$ of a field $L$.  Since \(L\) is a field, all left and right ideals in the skew left polynomial ring \(\lOre{L}{x}{\sigma}\) are principal. So greatest common left and right divisors and least common left and right multiples exist and they can be computed with the corresponding versions of the extended Euclidean algorithm, see e.g. \cite{Gomez:2014}. Concretely given \(f,g \in \lOre{L}{x}{\sigma}\), the least common left multiple of \(f\) and \(g\), denoted by \(\lclm{f,g}\), is the monic generator of \(\lOre{L}{x}{\sigma} f \cap \lOre{L}{x}{\sigma} g\), and the greatest common right divisor, denoted by \(\gcrd{f,g}\), is the monic generator of \(\lOre{L}{x}{\sigma} f + \lOre{L}{x}{\sigma} g\). Analogously, the least common right multiple is denoted by \(\lcrm{f,g}\) and the greatest common left divisor by \(\gcld{f,g}\). 

The theory developed in Section \ref{sec:constacyclic} can be applied to \(\mathcal{R} = \lOre{L}{x}{\sigma} / \langle x^n-1 \rangle\) with \(u = 1\). In this case, \(\mathcal{R} = \widehat{\mathcal{R}}\) and \((1,\sigma)\)-constacyclic codes are called \(\sigma\)-cyclic codes. Theorem \ref{dualu2=1} says that the dual of a \(\sigma\)-cyclic code is \(\sigma\)-cyclic. In this \(\sigma\)-cyclic setting the morphism \(\Theta\) becomes an involution, i.e.
\[
\Theta : \mathcal{R} \to \mathcal{R}, \qquad \sum_{i = 0}^{n-1}a_ix^i \mapsto \sum_{i = 0}^{n-1}\sigma^{-i}(a_i)x^{-i} = a_0 + \sum_{j=1}^{n-1} \sigma^j(a_{n-j}) x^j,
\]
satisfies \(\Theta^2 = \identity{\mathcal{R}}\). 

Skew Reed-Solomon codes are \(\sigma\)-cyclic codes generated by some special polynomials which we describe below. Let \(\{\alpha, \sigma(\alpha), \dots, \sigma^{n-1}(\alpha)\}\) be a normal basis of \(L/K\) where \(K = L^\sigma\) denotes de invariant subfield. Let \(\beta = \sigma(\alpha) \alpha^{-1}\). By \cite[Lemma 3.1]{GLNPGZ},
\begin{equation}\label{eq:fulldecomposition}
x^n-1 = \lclm{x-\beta, x-\sigma(\beta), \dots, x-\sigma^{n-1}(\beta)}.
\end{equation}
Let \(t\) such that \(2t < n\), \(\delta = 2t+1\). Let 
\[
g = \lclm{x-\beta, x-\sigma(\beta), \dots, x-\sigma^{\delta-2}(\beta)}.
\]
The code \(\mathcal{C} = \tovector(\mathcal{R}g)\) is called a skew RS code of designed distance \(\delta\). If \(k = n - \delta + 1\), \(\mathcal{C}\) is an \([n,k,\delta]\)-code over \(L\), hence it is MDS with respect to the Hamming distance (see \cite[Definition 2, Theorem 3.4]{GLNPGZ}). Skew Reed-Solomon codes can be efficiently decoded, see \cite{GLNPGZ,gln2017sugiyama}

We will prove that the dual of a skew Reed-Solomon code is again a skew Reed-Solomon code, and describe explicitly its generator polynomial as a least common left multiple of linear polynomials. Some preliminary results are needed.  Observe that $\sigma$ extends to an automorphism of $\mathcal{R}$, which acts on $x$  as the identity. 

\begin{lemma}\label{sigmaTheta}
The morphisms \(\Theta\) and \(\sigma\) commute, i.e. \(\Theta \sigma = \sigma \Theta\). 
\end{lemma}

\begin{proof}
It is a straightforward computation. 
\end{proof}

Let \(\gamma \in L\) such that $x-\gamma$ left divides $x^n-1$.  Then \((x - \gamma) \mathcal{R}\) is a maximal right ideal, and therefore \(\mathcal{R} \Theta(x - \gamma)\) is a maximal left ideal. Hence there exists \(\gamma' \in L\) such that 

\[
\mathcal{R}(x - \gamma') = \mathcal{R} \Theta(x - \gamma).
\]
In fact, since $\Theta(x-\gamma) = x^{-1} - \gamma$, we get that 
\[
\gamma'  = \sigma(\gamma)^{-1}.
\]

\begin{lemma}\label{fromhtoTheta(h)}
If \(h = \lcrm{\{x - \sigma^i(\gamma)~|~0 \leq i \leq k-1\}}\), then $\mathcal{R}\Theta(h) = \mathcal{R}h'$, where
\[ 
h' = \lclm{\{x - \sigma^i(\gamma')~|~0 \leq i \leq k-1\}}.
\]
\end{lemma}

\begin{proof}
Since \(h = \lcrm{\{x - \sigma^i(\gamma)~|~0 \leq i \leq k-1\}}\),
\[
h \mathcal{R} = \bigcap_{i=0}^{k-1} (x - \sigma^i(\gamma)) \mathcal{R} = \bigcap_{i=0}^{k-1} \sigma^i \left( (x - \gamma) \mathcal{R} \right).
\]
So, by Lemma \ref{sigmaTheta},
\[
\begin{split}
\mathcal{R} \Theta(h) &= \Theta \left( h \mathcal{R} \right) \\
&= \Theta \left( \bigcap_{i=0}^{k-1} \sigma^i \left( (x - \gamma) \mathcal{R} \right) \right) \\
&= \bigcap_{i=0}^{k-1} \Theta \left( \sigma^i \left( (x - \gamma) \mathcal{R} \right) \right) \\
&= \bigcap_{i=0}^{k-1} \sigma^i \left( \Theta \left( (x - \gamma) \mathcal{R} \right) \right) \\
&= \bigcap_{i=0}^{k-1} \sigma^i \left( \mathcal{R} \Theta (x - \gamma) \right) \\
&= \bigcap_{i=0}^{k-1} \sigma^i \left( \mathcal{R} (x - \gamma') \right) \\
&= \bigcap_{i=0}^{k-1} \mathcal{R} (x - \sigma^i (\gamma')) \\
&= \mathcal{R} \lclm{\{x - \sigma^i(\gamma')~|~0 \leq i \leq k-1\}}.
\end{split}
\]
\end{proof}

\begin{lemma}\label{rightleft}
Let \(\gamma \in L\) such that $(x-\gamma) \lclm{x-\sigma(\beta),\ldots,x-\sigma^{n-1}(\beta)}= x^n-1$. Then 
\[
\lcrm{x-\gamma, \ldots , x-\sigma^k(\gamma)} \lclm{x-\sigma^{k+1}(\beta),\ldots,x-\sigma^{n-1}(\beta)}=x^n-1
\]
for all \(0 \leq k \leq n-1\).
\end{lemma}

\begin{proof}
Denote $N_i=\lclm{\{x-\sigma^j(\beta)~|~j\neq i\}}$ for each $i=0,\dots , n-1$. 
Then $N_i$ has degree $n-1$ by \eqref{eq:fulldecomposition}. So, there exist  $\gamma_0,\gamma_1,\ldots ,\gamma_{n-1} \in L$ such that
\[
(x-\gamma_i)N_i=x^n-1
\]
for  $i=0,\ldots, n-1$. Let $\gamma=\gamma_0$. Then $(x-\gamma)N_0=x^n-1$, and therefore $(x-\sigma^k(\gamma))\sigma^k(N_0)=x^n-1$ for  $k=1,\ldots ,n-1$.
On the other hand, observe that $\sigma^k(N_0) = N_k$ for any $k=1,\ldots ,n-1$. Thus, for  $k=1,\ldots ,n-1$, $(x-\sigma^k(\gamma))N_k=(x-\gamma_k)N_k$, and then $\gamma_k=\sigma^k(\gamma)$ for all $k=1,\ldots ,n-1$. 

Let us now prove our thesis, i.e. 
\[
\lcrm{x-\gamma, \ldots , x-\sigma^k(\gamma)} \lclm{x-\sigma^{k+1}(\beta),\ldots,x-\sigma^{n-1}(\beta)}=x^n-1
\]
for  $k=0,\ldots ,n-1$, by induction on \(k\). The case $k=0$ is trivial. Assume that it holds for $j\leq k$ and we prove it for $k+1$. Firstly, the polynomial $\lcrm{x-\gamma, \ldots , x-\sigma^{k+1}(\gamma)}$ has degree $k+2$. Indeed, otherwise, by hypothesis,  $\lcrm{x-\gamma, \ldots , x-\sigma^{k}(\gamma)} = \lcrm{x-\gamma, \ldots , x-\sigma^{k+1}(\gamma)}$ and then 
$(x-\sigma^{k+1}(\gamma))p= \lcrm{x-\gamma, \ldots , x-\sigma^{k}(\gamma)}$. So, 
\[
(x-\sigma^{k+1}(\gamma))p \lclm{x-\sigma^{k+1}(\beta),\ldots,x-\sigma^{n-1}(\beta)} = x^n-1.
\]
In particular, this implies that 
\[
p \lclm{x-\sigma^{k+1}(\beta),\ldots,x-\sigma^{n-1}(\beta)} = N_{k+1}
\]
and, consequently, $x-\sigma^{k+1}(\beta)$ right divides $N_{k+1}$, which contradicts 
\eqref{eq:fulldecomposition}.

Now, by hypothesis, 
\[
\lcrm{x-\gamma, \ldots , x-\sigma^k(\gamma)} \lclm{x-\sigma^{k+1}(\beta),\ldots,x-\sigma^{n-1}(\beta)} = x^n-1
\]
so, applying $\sigma$,
\[
\lcrm{x-\sigma(\gamma), \ldots , x-\sigma^{k+1}(\gamma)} \lclm{x-\sigma^{k+2}(\beta),\ldots,x-\sigma^{n-1}(\beta),x-\beta} = x^n-1
\]
so
\[
\lcrm{x-\sigma(\gamma), \ldots , x-\sigma^{k+1}(\gamma)} q \lclm{x-\sigma^{k+2}(\beta),\ldots,x-\sigma^{n-1}(\beta)} = x^n-1
\]
for some polynomial $q$. That is, if $h$ is the monic polynomial verifying that
\[
h \lclm{x-\sigma^{k+2}(\beta),\ldots,x-\sigma^{n-1}(\beta)} = x^n-1,
\]
then $\lcrm{x-\sigma(\gamma), \ldots , x-\sigma^{k+1}(\gamma)}$ left divides $h$. On the other hand, the case $k=0$ provides that
\[
(x-\gamma) \lclm{x-\sigma(\beta),\ldots,x-\sigma^{n-1}(\beta)} = x^n-1
\] 
and therefore
$x-\gamma$ left divides $h$ as well. So $\lcrm{x-\gamma, x-\sigma(\gamma), \ldots , x-\sigma^{k+1}(\gamma)}$ left divides $h$. Since both have degree $k+2$, the result follows.
\end{proof}

\begin{theorem}\label{sRSdual}
Consider a skew RS code \(\mathcal{C} = \tovector(\R g)\), where 
\[
g = \lclm{x-\beta, x-\sigma(\beta), \dots, x-\sigma^{\delta-2}(\beta)}.
\]
If \(\gamma \in L\) is such that 
\[ (x-\gamma) \lclm{x-\sigma(\beta),\ldots,x-\sigma^{n-1}(\beta)}= x^n-1,
\] then $\mathcal{C}^\bot$ is the skew RS code generated by
\[
 \lclm{x-\sigma^{\delta}(\gamma)^{-1},  \dots, x-\sigma^{n}(\gamma)^{-1}}.
\]

\end{theorem}

\begin{proof}
By Lemma \ref{rightleft}
\[
\lcrm{x-\gamma, \ldots , x-\sigma^{n-\delta}(\gamma)} \lclm{x-\sigma^{n-\delta+1}(\beta),\ldots,x-\sigma^{n-1}(\beta)}=x^n-1,
\]
and, by applying \(\sigma^{\delta-1}\), we obtain
\[
\lcrm{x-\sigma^{\delta-1}(\gamma), \ldots , x-\sigma^{n-1}(\gamma)} \lclm{x-\beta,\ldots,x-\sigma^{\delta-2}(\beta)}=x^n-1.
\]
Therefore, $hg = x^n-1$, where $h = \lcrm{x-\sigma^{\delta-1}(\gamma), \ldots , x-\sigma^{n-1}(\gamma)}$. A standard argument, which uses that $x^n-1$ is central, proves that $gh = x^n -1$. 

By Corollary \ref{dualu2=1}, we get that $\mathcal{C}^\bot = \mathcal{R}\Theta(h)$. Lemma \ref{fromhtoTheta(h)} gives then that $\mathcal{R}\Theta(h)$ is generated by
\[
\lclm{x-\sigma(\sigma^{\delta -1}(\gamma))^{-1}, \ldots, \sigma^{n-\delta+1}(\sigma^{\delta-1}(\gamma))^{-1}},
\]
which finishes the proof. 
\end{proof}

We finish by proving that, as a consequence of Theorem \ref{sRSdual}, skew RS codes can be seen as evaluation codes. The right evaluation of a skew polynomial \(f = \sum_i f_i x^i \in \lOre{L}{x}{\sigma}\) by \(a \in L\) is the remainder of the right division of \(f\) by \(x - a\), i.e. the unique element \(f(a) \in L\) such that \(f(x) = q(x) (x-a) + f(a)\). Then
\[
f(a) = \sum_i f_i \norm{i}{a}
\]
where
\[
\norm{i}{a} = a \sigma(a) \dots \sigma^{i-1}(a).
\]

\begin{definition}
Let \(\overline{\alpha} = (\alpha_0, \alpha_1, \dots, \alpha_{m-1}) \in L^m\) and \(\overline{v} = (v_0, v_1, \dots, v_{m-1}) \in (L\setminus\{0\})^m\). The skew Generalized Evaluation code associated to \((\overline{\alpha},\overline{v})\) is 
\[
\operatorname{sGE}_{(\overline{\alpha},\overline{v})} = \left\{ (v_0 f(\alpha_0), \dots, v_{m-1} f(\alpha_{m-1})) ~|~ f \in \lOre{L}{x}{\sigma}, \deg f < k \right\}.
\]
\end{definition}
It is straightforward to check that \(\operatorname{sGE}_{(\overline{\alpha},\overline{v})}\) is an \(L\)--linear code. In fact a generator matrix for it is
\begin{multline*}
\left(\begin{matrix}
v_0 & v_1 & \cdots & v_{m-1} \\
v_0 \alpha_0 & v_1 \alpha_1 & \cdots & v_{m-1} \alpha_{m-1} \\
\vdots & \vdots & \ddots & \vdots \\
v_0 \norm{k-1}{\alpha_0} & v_1 \norm{k-1}{\alpha_1} & \cdots & v_{n-1}\norm{k-1}{\alpha_{m-1}}
\end{matrix}\right) \\
= 
\left(\begin{matrix}
1 & 1 & \cdots & 1 \\
\alpha_0 & \alpha_1 & \cdots & \alpha_{m-1} \\
\vdots & \vdots & \ddots & \vdots \\
\norm{k-1}{\alpha_0} & \norm{k-1}{\alpha_1} & \cdots & \norm{k-1}{\alpha_{m-1}}
\end{matrix}\right) 
\left(\begin{matrix}
v_0 & 0 & \cdots & 0 \\
0 & v_1 & \cdots & 0 \\
\vdots & \vdots & \ddots & \vdots \\
0 & 0 & \cdots & v_{m-1}
\end{matrix}\right)
\end{multline*}

This definition is an extension of \cite[Definition 9]{Liu/etal:2015} to arbitrary fields, including therefore the convolutional case when \(L = \field{}(t)\). A different approach to evaluation codes, where the norms are replaced by powers of the automorphism, can be found in \cite{Augot/etal:2013,Augot/etal:2017}.

Let us now prove that a skew RS code is also a sGE code. 

\begin{theorem}\label{sRSevaluation}
Let \(\mathcal{C} = \tovector(\mathcal{R} g)\) be a skew RS code where
\[
g = \lclm{x-\beta, \dots, x - \sigma^{\delta-2}(\beta)}.
\]
Then there exist \(\mu, \nu \in L\setminus\{0\}\) such that 
\(
\mathcal{C} = \operatorname{sGE}_{(\overline{\mu},\overline{\nu})},
\)
where \(\overline{\mu} = (\mu, \sigma(\mu), \dots, \sigma^{n-1}(\mu))\) and \(\overline{\nu} = (\nu, \sigma(\nu), \dots, \sigma^{n-1}(\nu))\).
\end{theorem}

\begin{proof}
By Theorem \ref{sRSdual} there exists \(\mu \in L\) such that \(\mathcal{C}^\bot = \tovector(\mathcal{R} g')\) where 
\[
g' = \lclm{x - \mu, \dots, x - \sigma^{n-\delta}(\mu)}.
\]
Since \(\sigma^{i}(\mu)\) is a right root of \(g'\) for \(0 \leq i \leq n-\delta\), it follows that a parity check matrix of \(\mathcal{C}^\bot\) is 
\[
\left(\begin{matrix}
1 & 1 & \dots & 1 \\
\mu & \sigma(\mu) & \dots & \sigma^{n-\delta}(\mu) \\
\vdots & \vdots & \ddots & \vdots \\
\norm{n-1}{\mu} & \norm{n-1}{\sigma(\mu)} & \dots & \norm{n-1}{\sigma^{n-\delta}(\mu)}
\end{matrix}\right).
\]
Since \(\mu\) is also a left root of \(x^n-1\), it follows that \(\norm{n}{\mu} = 1\) and, by Hilbert's 90 Theorem, there exists \(\nu \in L\) such that \(\mu = \sigma(\nu)\nu^{-1}\). So, up to multiply each column by the corresponding scalar, a new parity check matrix for \(\mathcal{C}^\bot\) is 
\[
H = \left(\begin{matrix}
\nu & \sigma(\nu) & \cdots & \sigma^{n-\delta}(\nu) \\
\sigma(\nu) & \sigma^2(\nu) & \cdots & \sigma^{n-\delta+1}(\nu) \\
\vdots & \vdots & \ddots & \vdots \\
\sigma^{n-1}(\nu) & \nu & \cdots & \sigma^{n-\delta-1}(\nu)
\end{matrix}\right).
\] 
Therefore \(\mathcal{C}\) is generated by the rows of the matrix
\[
\transpose{H} = \left(\begin{matrix}
\nu & \sigma(\nu) & \cdots & \sigma^{n-1}(\nu) \\
\sigma(\nu) & \sigma^2(\nu) & \cdots & \nu \\
\vdots & \vdots & \ddots & \vdots \\
\sigma^{n-\delta}(\nu) & \sigma^{n-\delta+1}(\nu) & \cdots & \sigma^{n-\delta-1}(\nu)
\end{matrix}\right).
\]
We have
\[
\begin{split}
\transpose{H} &= \left(\begin{matrix}
\nu & \sigma(\nu) & \cdots & \sigma^{n-1}(\nu) \\
\sigma(\nu) & \sigma^2(\nu) & \cdots & \nu \\
\vdots & \vdots & \ddots & \vdots \\
\sigma^{n-\delta}(\nu) & \sigma^{n-\delta+1}(\nu) & \cdots & \sigma^{n-\delta-1}(\nu)
\end{matrix}\right) \\
&= \left(\begin{matrix}
1 & 1 & \cdots & 1 \\
\sigma(\nu)\nu^{-1} & \sigma^2(\nu)\sigma(\nu)^{-1} & \cdots & \nu \sigma^{n-1}(\nu)^{-1} \\
\vdots & \vdots & \ddots & \vdots \\
\sigma^{n-\delta}(\nu)\nu^{-1} & \sigma^{n-\delta+1}(\nu)\sigma(\nu)^{-1} & \cdots & \sigma^{n-\delta-1}(\nu)\sigma^{n-1}(\nu)^{-1}
\end{matrix}\right) \\
&\qquad \cdot
\left(\begin{matrix}
\nu & 0 & \cdots & 0 \\
0 & \sigma(\nu) & \cdots & 0 \\
\vdots & \vdots & \ddots & \vdots \\
0 & 0 & \cdots & \sigma^{n-1}(\nu)
\end{matrix}\right) \\
&= \left(\begin{matrix}
1 & 1 & \cdots & 1 \\
\mu & \sigma(\mu) & \cdots & \sigma^{n-1}(\mu) \\
\vdots & \vdots & \ddots & \vdots \\
\norm{n-\delta}{\mu} & \norm{n-\delta}{\sigma(\mu)} & \cdots & \norm{n-\delta}{\sigma^{n-1}(\mu)}
\end{matrix}\right) \\
&\qquad \cdot
\left(\begin{matrix}
\nu & 0 & \cdots & 0 \\
0 & \sigma(\nu) & \cdots & 0 \\
\vdots & \vdots & \ddots & \vdots \\
0 & 0 & \cdots & \sigma^{n-1}(\nu)
\end{matrix}\right),
\end{split}
\]
hence \(\mathcal{C} = \operatorname{sGE}_{(\overline{\gamma},\overline{\nu})}.\)
\end{proof}

\bibliographystyle{ams}

\end{document}